%% file: paper.tex
\documentclass[runningheads,a4paper]{llncs}
\usepackage[T1]{fontenc}
\usepackage{graphicx}
\usepackage{xcolor} 
\usepackage{glossaries-extra} 
\usepackage{amsfonts} 
\usepackage{ragged2e} 
\usepackage{booktabs} 
\usepackage{colortbl} 
\usepackage{tabularx} 
\usepackage{algorithm} 
\usepackage{algpseudocode}
\usepackage[colorlinks=true,allcolors=blue]{hyperref}
\newcommand{\leftr}{\stackrel{R}{\leftarrow}}

\AtBeginDocument{%
  \mathchardef\stdcomma=\mathcode`,
  \mathcode`,="8000
}
\begingroup\lccode`~=`, \lowercase{\endgroup\def~}{\stdcomma\,}

\newcolumntype{L}[1]{>{\raggedright\let\newline\\\arraybackslash\hspace{0pt}}m{#1}} 

\usepackage{listings}
\lstdefinelanguage{json}{
    basicstyle=\normalfont\ttfamily,
    numbers=left,
    numberstyle=\scriptsize,
    breaklines=true,
    frame=lines,
    backgroundcolor=\color{gray!10},
    showstringspaces=false,
    string=[db]{"},
    stringstyle=\color{green!50!black},
    morestring=[s][\color{black}]{\ \ "}{":},
    keywordstyle=\color{blue},
    keywords={true,false,null},
    literate=
     *{0}{{{\color{red}0}}}{1}
      {1}{{{\color{red}1}}}{1}
      {2}{{{\color{red}2}}}{1}
      {3}{{{\color{red}3}}}{1}
      {4}{{{\color{red}4}}}{1}
      {5}{{{\color{red}5}}}{1}
      {6}{{{\color{red}6}}}{1}
      {7}{{{\color{red}7}}}{1}
      {8}{{{\color{red}8}}}{1}
      {9}{{{\color{red}9}}}{1}
      {.}{{{\color{red}.}}}{1}
      {:}{{{\color{gray}{:}}}}{1}
      {,}{{{\color{gray}{,}}}}{1}
      {\{}{{{\color{gray}{\{}}}}{1}
      {\}}{{{\color{gray}{\}}}}}{1}
      {[}{{{\color{gray}{[}}}}{1}
      {]}{{{\color{gray}{]}}}}{1},
}

\setabbreviationstyle[acronym]{long-short}
\newacronym{ssi}{SSI}{Self-Sovereign Identity}
\newacronym[plural=VCs,firstplural=Verifiable Credentials (VCs)]{vc}{VC}{Verifiable Credential}
\newacronym[plural=RPs,firstplural=Relying Parties (RPs)]{rp}{RP}{Relying Party}
\newacronym{ocsp}{OCSP}{Online Certificate Status Protocol}
\newacronym{crl}{CRL}{Certificate Revocation List}
\newacronym[plural=AMQs,firstplural=approximate membership query data structures (AMQs)]{amq}{AMQ}{approximate membership query data structure}
\newacronym{ebsi}{EBSI}{European Blockchain Service Infrastructure}
\newacronym{oid4vc}{OID4VC}{OpenID for Verifiable Credentials}
\newacronym{anoncreds}{AnonCreds}{Anonymous Credentials}
\newacronym[plural=DIDs,firstplural=Decentralized Identifiers (DIDs)]{did}{DID}{Decentralized Identifier}
\newacronym{mse}{MSE}{mean squared error}
\newacronym{mae}{MAE}{mean absolute error}
\newacronym{bsl}{BSL}{Bitstring Status List}

\begin{document}
\title{CRSet: Private Non-Interactive Verifiable Credential Revocation}
%
%
\author{Felix Hoops \and
Jonas Gebele \and
Florian Matthes}
\authorrunning{F. Hoops et al.}
%
\institute{Technical University of Munich, Munich, Germany\\
\email{\{felix.hoops,jonas.gebele,matthes\}@tum.de}}
\maketitle              
\begin{abstract}
Like any digital certificate, Verifiable Credentials (VCs) require a way to revoke them in case of an error or key compromise.
Existing solutions for VC revocation, most prominently Bitstring Status List, are not viable for many use cases because they may leak the issuer's activity, which in turn leaks internal business metrics. For instance, staff fluctuation through the revocation of employee IDs.
We identify the protection of issuer activity as a key gap and propose a formal definition for a corresponding characteristic of a revocation mechanism.
Then, we introduce CRSet, a non-interactive mechanism that trades some space efficiency to reach these privacy characteristics. For that, we provide a proof sketch.
Issuers periodically encode revocation data and publish it via Ethereum blob-carrying transactions, ensuring secure and private availability.
Relying Parties (RPs) can download it to perform revocation checks locally.
Sticking to a non-interactive design also makes adoption easier because it requires no changes to wallet agents and exchange protocols.
We also implement and empirically evaluate CRSet, finding its real-world behavior to match expectations.
One Ethereum blob fits revocation data for about 170,000 VCs.

\end{abstract}
\section{Introduction}

\input{content/01_introduction}


\section{Design Considerations}
\label{sec:system}
\input{content/10_challenges}

\section{Proposed Revocation Mechanism}
\label{sec:mechanism}
\input{content/25_mechanism}

\section{Issuer Privacy Model}
\label{sec:proof}
\input{content/27_proof}

\section{Implementation and Evaluation}
\label{sec:evaluation}
\input{content/30_eval}

\section{Conclusion}
\label{sec:conclusion}
\input{content/40_conclusion}

\begin{credits}
\subsubsection{\ackname} We thank the Ethereum Foundation for their support in the form of an Ethereum Academic Grant under reference number FY24-1545.
We thank Evan Christopher, Zexin Gong, Chan-Young Lee, and Natalia Milanova for their work on prototype libraries.

\end{credits}

\appendix

%
%
%
\bibliographystyle{splncs04}
\bibliography{bibliography}

\end{document}

%% file: content/01_introduction.tex
\gls{ssi} \cite{allen2016path} is gaining traction as a user-centric approach to digital identity, offering tangible benefits in security and efficiency. Government-led initiatives like the European Blockchain Service Infrastructure and the U.S. Department of Homeland Security are already working on \gls{ssi}-based systems for issuing digital credentials like diplomas \cite{europaEBSI} and immigration records \cite{dhsWhoYouAre}.

\glspl{vc} are digitally signed, structured data objects used to represent attestable claims such as identity documents, professional licenses, or organizational roles \cite{w3VerifiableCredentials}. As with any credential system, the ability to revoke previously issued \glspl{vc} is essential to correct errors, respond to key compromise, or reflect changes in status. Revocation, however, remains poorly addressed in the \gls{ssi} ecosystem \cite{mazzocca2024survey,satybaldyTaxonomyChallengesSelfSovereign2024,schardongSelfSovereignIdentitySystematic2022}. No existing solution protects issuers from leaking operational information such as the number or timing of revocations. This missing capability poses a significant barrier to adoption. Organizations issuing revocable \glspl{vc} today risk exposing sensitive internal data. Most mechanisms reveal each revocation publicly, enabling observers to infer patterns such as employee departures when \glspl{vc} are used for access or staff credentials.

Conventional mechanisms from public-key infrastructures like \gls{crl} or the \gls{ocsp} are equally problematic. They either expose aggregate revocation counts \cite{rfc5280} or leak usage metadata during verification \cite{rfc2560}, making them unsuitable for \gls{ssi}’s privacy requirements.
\gls{ocsp} stapling is closest, but it still leaks when a credential is intended to be used.
The most widely used \gls{ssi}-native mechanisms also fall short. The W3C \gls{bsl} \cite{w3BitstringStatus} compresses revocation data efficiently but makes individual revocations observable. Hyperledger \gls{anoncreds} Revocation \cite{hlAnoncreds} offers subject anonymity, but still discloses issuer-level revocation information through public accumulator data.

\begin{table*}[!t]
\centering
\caption{Comparison of Credential Revocation Mechanisms for SSI}
\renewcommand{\arraystretch}{1.2}
\setlength{\tabcolsep}{3pt}
\begin{tabular}{p{0.24\textwidth}p{0.18\textwidth}p{0.12\textwidth}p{0.12\textwidth}p{0.12\textwidth}p{0.12\textwidth}}
\toprule
\textit{Mechanism} & \textit{Revocation \mbox{Data}} & \textit{Subject \mbox{Privacy}} & \textit{RP \mbox{Privacy}} & \textit{Issuer \mbox{Privacy}} & \textit{Non-Interactive} \\ \midrule
\mbox{Bitstring} \mbox{Status List} & $< 1$ bit/cap. & PIR\footnotemark & No & Rate\footnotemark & \textbf{Yes} \\
\mbox{AnonCreds} \mbox{Rev.} & $>12$ B/cap. & PIR & \textbf{Yes} & No & No \\
Prevoke \cite{manimaran2024prevoke} & ? & \textbf{Yes} & \textbf{Yes} & Rate & No \\
Evoke \cite{mazzocca2024evoke} & $\geq 32$ B/rev.& \textbf{Yes} & Depends\footnotemark & Rate & No \\
\mbox{Sitouah} \mbox{et al. \cite{sitouah2024untraceable}} & $\geq 2$ B/valid\footnotemark & \textbf{Yes} & \textbf{Yes} & Rate & No \\
\rowcolor{lightgray}
CRSet & $\approx 6$ bit/cap. & \textbf{Yes} & \textbf{Yes} & \textbf{Yes} & \textbf{Yes} \\
\bottomrule
\end{tabular}%
\label{tab:comparison}
\end{table*}

\footnotetext[1]{Private Information Retrieval}
\footnotetext[2]{Revocation Rate Hiding}
\footnotetext[3]{Depends on the exact workings of the used gossip protocol.}
\footnotetext[4]{Practical measurement from paper for parameterization with at least $10^5$ capacity.}

We introduce Credential Revocation Set (CRSet), a revocation mechanism for non-anonymous credentials that require strong protection of issuer activity. CRSet uses Bloom filter cascades \cite{larisch2017crlite} combined with fixed-size padding and a regular publishing schedule to conceal both the total number and frequency of revocations. It enables revocation checks without active participation of issuers, and without exposing metadata, ensuring privacy for issuers, \glspl{rp}, and credential holders alike.
No prior revocation mechanism we are aware of protects absolute and relative issuer activity to the level that CRSet does.
We provide a high-level comparison of CRSet and other existing \gls{ssi}-native revocation mechanisms in Table~\ref{tab:comparison}.

Our contributions are:

\begin{itemize}
    \item We provide a structured analysis of privacy concerns in credential revocation mechanism within \gls{ssi} ecosystems, with a focus on issuers. We formalize issuer privacy using a game-based definition of activity indistinguishability.
    \item We introduce CRSet, a space-efficient, non-interactive status revocation mechanism that improves metadata privacy for issuers, \glspl{rp}, and subjects. We prove that CRSet satisfies the defined issuer privacy properties.
    \item We implement both a standalone prototype of the CRSet data structure and a complete end-to-end demonstrator. Our empirical evaluation confirms CRSet’s privacy properties and analyzes performance and scalability.
\end{itemize}

The remainder of this paper is organized as follows: 
First, we perform a brief survey of related research work in Section~\ref{sec:related}.
Next, in Section~\ref{sec:system}, we describe stakeholders and privacy concerns before setting goals and reflecting on possible underlying data structures.
Then, we present the CRSet design in Section~\ref{sec:mechanism}, covering data structure and publishing.
In Section~\ref{sec:proof}, we introduce a formal definition for an issuer's privacy along with notation and a proof sketch showing that CRSet meets that level of privacy.
In Section~\ref{sec:evaluation}, we implement and empirically evaluate CRSet, before coming to a conclusion in Section~\ref{sec:conclusion}.

\section{Related Work}
\label{sec:related}

The arguably most notable revocation mechanism for \glspl{vc} is the W3C's Bitstring Status List (BSL) \cite{w3BitstringStatus}. As the name suggests, it is a bitstring of pre-allocated length. Every issued \gls{vc} is associated with a unique index, which can be set to one to revoke it. Since most values will typically be zero, GZIP compression is applied to great effect before publicly hosting it on the internet.
\glspl{rp} fetch a full status list for a local revocation check, which creates some level of private information retrieval.
However, when checking a \gls{vc}, the \gls{rp} will notify the corresponding issuer that it is doing so when fetching from the issuer's server.
\gls{bsl} also reveals every individual revocation an issuer makes.

The only other practically relevant mechanism is one tailored to Hyperledger's \gls{anoncreds} credential standard \cite{hlAnoncreds}.
\gls{anoncreds}' focus is to enable subjects to interact with \glspl{rp} while staying anonymous, which requires anonymous revocation checks. An approach based on a cryptographic accumulator is used. Subjects have to interactively construct a zero-knowledge proof (ZKP) of inclusion in that accumulator, which can be verified by an \gls{rp}. Inclusion signifies that a \gls{vc} has not been revoked.
But once again, this mechanism has privacy issues.
The ZKP relies on a public set of prime numbers (i.e., a tails file), where one prime is associated with each issued credential. Due to the immense size, this file is usually on a web server owned by the issuer, making access trackable.
In addition, the ZKP construction requires knowing which of these primes belong to revoked credentials. That information is also public in the form of what is effectively a bitstring using a one to indicate revocation. Although the exact implementation can vary (e.g., expressed as deltas), the information is present and allows the same kinds of insights into the issuer's operations that \gls{bsl} does.

A few less practically established revocation mechanisms for \gls{ssi} ecosystems exist. Some of them predate the notion of \gls{ssi}, but consider credentials in the same spirit.
Specifically, there are several works addressing revocation for anonymous credentials \cite{lueks2017fast,acar2011revocation} similar to \gls{anoncreds}, which builds on \cite{camenisch2001efficient}.
Sitouah et al. \cite{sitouah2024untraceable} propose an interactive mechanism based on a Merkle tree accumulator, which limits how long an \gls{rp} can trace the credential status.
Chotkan et al. \cite{chotkan2022distributed} focus on creating a revocation mechanism that is as decentralized as possible, even avoiding the logical decentralization of distributed ledger technology.
Mazzocca et al. \cite{mazzocca2024evoke} develop a revocation mechanism using an ECC-based accumulator, which is suitable for limited hardware power in IoT scenarios. 
Manimaran et al. \cite{manimaran2024prevoke} propose an accumulator-based mechanism and consider an issuer's privacy, but focus on revocation rate only.

Bloom filters have been introduced by Bloom in \cite{bloom1970space} and have established themselves as the default probabilistic membership query data structure.
Layering multiple Bloom filters has been explored in several works for different applications related to indexing \cite{wang2009casab,salikhov2014using,rozov2014fast,mousavi2019constructing}.
Larisch et al. \cite{larisch2017crlite} have specifically applied Bloom filter cascades to the problem of certificate revocation for TLS.

%% file: content/10_challenges.tex
We start out by describing the kind of \gls{ssi} ecosystem we design for, set goals, and finally motivate our choice of data structure.

\begin{figure}[!t]
    \centering
    \includegraphics[width=0.75\textwidth]{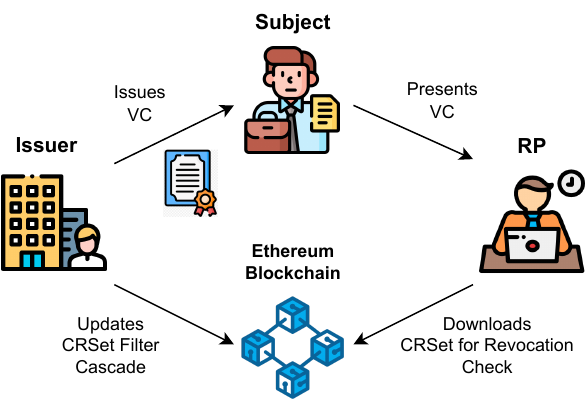}
    \caption{Stakeholder Interactions in a Typical SSI Ecosystem Using CRSet}
    \label{fig:trust-route}
\end{figure}

\subsection{System and Stakeholders}

A typical \gls{ssi} ecosystem involves three types of stakeholders: issuers, subjects, and \glspl{rp}.
Subjects can receive \glspl{vc} from issuers (e.g., employee IDs) and then present these \glspl{vc} to various \glspl{rp} at will (e.g., to gain access to a service).
To maintain the trust of \glspl{rp}, issuers must maintain record of their issued \glspl{vc} and publish some form of revocation data that invalidates \glspl{vc} when they are no longer accurate (e.g., employee leaves company). \glspl{rp} have to access this data whenever they validate a \gls{vc}.
In Figure~\ref{fig:trust-route}, we depict a concrete version of this ecosystem using our proposed revocation mechanism.

Each stakeholder may act as an adversary attempting to infer sensitive metadata.
For instance, issuers could monitor presentation behavior if they are involved in revocation checks (e.g., hosting data on a web server), violating the privacy of subject and \gls{rp}.
And anyone, even external actors, could access public revocation data and attempt to infer issuance or revocation events, possibly revealing metadata about an issuer (e.g., employee churn).
Therefore, the system must mitigate threats to two core categories of information: \textit{presentation activity}, which concerns metadata from presentation of \glspl{vc}, and \textit{issuer activity}, which covers details about issuance and revocation activity.

\subsection{Goals}
\label{sec:goals}

Issuers may have vastly different issuance and revocation volumes, making it essential that the status mechanism features \textbf{parameterizable capacity} regarding maximum capacity of an instance (e.g., one status list) to any size fit for the issuer's use case. In addition, this gives issuers the ability to make their own trade-offs between data size and the timeframe that one revocation mechanism instance is supposed to cover.

Next, we want to provide \textbf{confidentiality of issuer activity}.
When analyzing a single snapshot of revocation data from an issuer, an adversary must not be able to infer issuer metadata. Specifically, that is the number of revoked or issued \glspl{vc} that is covered by the revocation data. This implies that an adversary can also not infer the revocation rate.
Further, even when having access to a possibly infinite series of subsequent versions of revocation data from an issuer, it must still not be feasible for an adversary to extract issuer metadata. That covers relative and absolute numbers for revoked or issued \glspl{vc}.

Further, the presentation process between a subject and an \gls{rp} must be private between these two parties, creating \textbf{confidentiality of presentation activity}. That also includes metadata. No one, including the issuer of the presented \gls{vc}, must be able to know that a \gls{vc} from said issuer was presented. This implies that it is also impossible to infer who the \gls{rp} or the subject is, as the interaction should not be detectable in the first place.

Finally, to facilitate adoption, the credential revocation mechanism must provide good \textbf{interoperability}.
It must be simple to implement and as compatible with existing systems as possible. Software used by issuers for issuing and \glspl{rp} for verifying must naturally always be aware of the revocation mechanism. But user agents, commonly known as wallets, do not necessarily need to.
Thus, we aim to design a mechanism that requires no new interaction from a user agent, ensuring that existing ones can still be used without modification. It also entails that they can be more light-weight, as they do not need to perform any of the potentially expensive cryptographic operations associated with interactive revocation mechanisms.
That means that all relevant information for a revocation check must be compatible with the existing W3C \gls{vc} standard and contained in its \texttt{credentialStatus} field \cite{w3VerifiableCredentials}.

\subsection{Data Structure Considerations}

Bloom filters \cite{bloom1970space} are arguably the most widely-used \gls{amq}. They consist of a bitstring of length $m$ and, for the scope of this paper, just one hash function $h: \{0,1\}^* \rightarrow \mathbb{N}_m$.
An element $d$ can be inserted by setting the bit at position $h(d)$ to one.
Testing if an element $d'$ is a member requires checking if the bit at $h(d')$ is set to one. Due to possible hash collisions, this test may result in false positives but never false negatives.

Filter cascades based on an \gls{amq} offer a way to eliminate the false positives on a known subset (e.g., issued \gls{vc} IDs) of a universe (e.g., all possible \gls{vc} IDs) by encoding the false positives of an \gls{amq} in a subsequent \gls{amq} stopping when no more occur.
This is space efficient and suitable for revocation \cite{larisch2017crlite}.

To minimize total data size, we have also considered filter cascades based on other types of possibly more efficient filters, specifically Cuckoo filters \cite{fan2014cuckoo}, XOR filters \cite{grafXorFiltersFaster2020}, and Binary Fuse filters \cite{graf2022binary}. 
However, we did not find them to be suitable for this specific application. Cuckoo filters rely on buckets that have fingerprints (i.e., hashes) explicitly stored within, allowing an adversary to simply count them.
XOR and Binary Fuse filters are not freely configurable with respect to a desired false positive rate $p$. For the purpose of building a space-efficient filter cascade, the available false positive rates of $0.4\%$ and $0.002\%$ are too small. It has previously been shown that for a Bloom filter cascade, a false positive rate of about $50\%$ is close to optimal \cite{larisch2017crlite}. Since these filters share a bit length per element of $\gamma \log_2(1/p)$ \cite{graf2022binary}, only distinguished by the constant factor $\gamma$, we assume that the optimal cascade parameters are similar, making any resulting cascades considerably less space-efficient.

%% file: content/25_mechanism.tex
In this section, we present CRSet. We start by explaining how \glspl{vc} are represented in the Bloom filter cascade and how we pad the data. Next, we address how the cascade is published on the Ethereum blockchain, followed by the algorithms for creating and reading a CRSet instance.
Finally, we reflect on our initial goals.

\subsection{Identifiers and Hashing}

Every \gls{vc} is randomly assigned a revocation ID, which we will refer to as just “ID”, from the large universe $U=\{0,1\}^{256}$. It is entirely unrelated to the \gls{vc}'s \texttt{id} field \cite{w3VerifiableCredentials}. Using random IDs prevents statistical attacks \cite{reviriegoAttackingPrivacyApproximate2023}. At any discrete point in time $i$, an issuer knows the set of IDs of his issued \glspl{vc} $C_i \subset U$. Further, he knows which IDs he has revoked $R_i \subseteq C_i$ and which are still valid $V_i \subseteq C_i$, where $R_i \cap V_i = \emptyset$. To simplify notation, we use these sets without index when referring to just one point in time.
Based on these sets of IDs, the Bloom filter cascade is created and published.

Usually, Bloom filters tend to use a hash function that is as fast as possible.
But due to our privacy goals, we mandate that the filters use a cryptographic hash function. This ensures that it is not possible to efficiently reverse any hashes, which could otherwise make it possible for an adversary to obtain at least a set of possible IDs based on the set bits in a filter.

\subsection{Padding}
\label{sec:padding}

Any adversary could infer the number of elements in the filter cascade based on the metadata of the Bloom filters making up the cascade, violating our goals. To combat this, we introduce padding. It is randomly chosen with the same process the IDs are.
The padded sets $\hat R$ and $\hat V$ are created such that $R \subseteq \hat R$, $V \subseteq \hat V$, and $\hat R \cap \hat V = \emptyset$.
Because it is also trivial to infer the ratio of included and excluded elements in a filter cascade, the padding size must be a fixed ratio to avoid leaking information.
We propose to adopt a fixed ratio of $n_{max} = |\hat V| = (1/2) |\hat R|$. This ensures that an issuer can use their cascade for a long time and still always emergency-revoke all valid certificates, which requires capacity among the revoked IDs. Since revoked IDs are padded to a larger size than the valid ones, they have to be the exclusions in a cascade to minimize data size \cite{larisch2017crlite}.

For any year $y$ (s.t. $y>t$), the total number of issued \glspl{vc} is defined as $|C_y| = v \sum_{i=y-t}^{y} \delta^i$ in the long-term, where \glspl{vc} have begun to expire, $v\in \mathbb{Z}$ is the yearly issuance volume, $x\in (0,1)$ the revocation percentage based on $v$, $\delta \in (1,2]$ the year-on-year increase of $v$, and $t\in \mathbb{Z}$ the lifetime of a \gls{vc}.
The minimum needed capacity for valid and revoked \glspl{vc} at a point in time $y$, ensuring $|\hat R|$ is large enough to enable total revocation, is:

\begin{equation*}
\begin{split}
    |\hat V| &\geq (1-x) * |C_y| = (1-x)v \sum_{i=y-t}^{y} \delta^i\\
    |\hat R| &\geq  (x+(1-x)) * |C_y| = |C_y| = \frac{1}{1-x} |\hat V|
\end{split}
\end{equation*}

\noindent
Assuming long-term revocation rates never exceed $50\%$, we fix the ratio at the worst case $x=1/2$, and arrive at $|\hat V| = (1/2) |\hat R|$.

\subsection{Publishing Mechanism}

The publishing mechanism is the main influence on choosing CRSet capacity $n_{max}$.
To maintain perfect deniability, an issuer must update, i.e., recreate and publish, the filter cascade in regular intervals, regardless of whether there are new revocations (e.g., every 24 hours).
The cascade should be published on a decentralized storage system, for which we use the Ethereum blockchain's blob-carrying transactions \cite{eip4844} due to their comparatively low cost for writing data to a public, distributed system. 
To optimize space usage on blobs, the padded size $n_{max}$ should be chosen such that the total bit size of the cascade is a multiple of the blob size.
Building on cascade size calculations from \cite{larisch2017crlite}, for one blob, CRSet can support $n_{max} \approx 128KB / 6b \approx 170.666$ \glspl{vc}.
Blobs are officially kept for a limited time \cite{eip4844}, which works well with regular updates where only the latest one is relevant for revocation checks.

\begin{figure}[t]
    \centering
    \begin{lstlisting}[language=json]
{ "credentialStatus": {
    "id": "eip155:1:0x32...53:dd...d3",
    "type": "CRSetEntry"
}}
    \end{lstlisting}
    \caption{Example of a Credential Status Entry for CRSet. The hexadecimal string segments have been shortened to improve readability.} 
    \label{fig:credentialStatus}
\end{figure}

The \gls{vc} specification defines the format of a credential status entry \cite{w3VerifiableCredentials}. It must be called \texttt{credentialStatus} and contain an \texttt{id} property of type URI, and a \texttt{type} property identifying the revocation mechanism, where we use \texttt{"CRSetEntry"}.
The \texttt{id} field must point to the published cascade.
We use the Chain Agnostic Improvement Protocol 10 (CAIP-10) \cite{caip10spec} as a standardized way to refer to the blockchain and the account on which it is published. Then we concatenate the revocation ID.
An example of a complete status section can be seen in Figure~\ref{fig:credentialStatus}.

\subsection{Algorithms and Performance}
\label{sec:algos}

\begin{algorithm}[t]
\caption{Creating a CRSet Filter Cascade.}
\label{alg:build}
\begin{algorithmic}[1]
\Require Valid IDs $V$, revoked IDs $R$, padding size $n_{max}$
\Ensure A padded CRSet Bloom filter $cascade$ with its $salt$
\State $(cascade, salt) \gets ([], rndElement( \{0,1\}^{256} ))$
\State $(p_0, p) \gets (sqrt(p)/2, 1/2)$ \Comment{Target false positive rates.}
\State $included \gets padWithRandomIDs(V, n_{max})$ \Comment{Working set for included IDs.}
\State $excluded \gets padWithRandomIDs(R, 2 n_{max})$ \Comment{Working set for excluded IDs.}
\For{$i \gets 0; |included| > 0; i \gets i+1$} \Comment{Eliminate false positives.}
    \State $cascade[i] \gets bloom(\{ id \| i \| salt ~|~ id \in excluded \},\textbf{if }i=0\textbf{ then }p_0\textbf{ else }p,1)$
    \State $fps \gets \{ id ~|~ id \in excluded \land id \| i \| salt \in cascade[i] \}$ \Comment{False Positives.}
    \State $excluded \gets included$
    \State $included \gets fps$
\EndFor
\State \textbf{return} $(cascade, salt)$
\end{algorithmic}
\end{algorithm}

When creating a new CRSet filter cascade, an issuer must choose an $n_{max}$ and apply Algorithm~\ref{alg:build}.
This same algorithm is also used for updates, which must be complete recreations in order to avoid making deltas quantifiable from the outside.
The algorithm generates random IDs for padding and subsequently creates the filter cascade.
Intuitively, the cascade is created to encode $V$ while not accidentally encoding elements from $R$ because of hash collisions.
It is created such that the first filter contains $V$. Every subsequent filter contains the relevant false positives of the filter before it until no more exist.
The individual Bloom filter initializations require three standard parameters: the set of entries (and implicitly the size of that set), the expected false positive probability, and the number of used hash functions.
We are using the two-probability strategy from \cite{larisch2017crlite} and thus have a $p_0$ for just the first filter, and a $p$ for all subsequent ones. The ideal number of hash functions in that case is one.

To ensure that cascades behave similar to expectation, which assumes true randomness, we vary the hashing between filters by introducing the level index as a deterministic salt.
We also generate a random salt for the entire cascade, which fulfills two purposes. First, it helps to make the filters more resilient against brute-force attacks. Second, it provides some randomness to the probabilistic process that is filter creation. In the event of it taking too long to converge toward zero false positives, the process can simply be restarted with the same inputs, which generates a new salt, making it likely to terminate correctly.

\begin{algorithm}[t]
\caption{Checking if an ID Is Unrevoked}
\label{alg:test}
\begin{algorithmic}[1]
\Require A \gls{vc}'s revocation $id$, and a filter $cascade$ with its $salt$
\Ensure \texttt{True} if the $id$ is in the $cascade$ and thus valid
\For{$i \gets 0; i < length(cascade); i \gets i+1$}
    \If{$id \| i \| salt ~\notin~ cascade[i]$}
        \State \textbf{return} $i \bmod 2 = 1$ \Comment{Even levels confirm cascade membership.}
    \EndIf \Comment{Odd levels refute cascade membership.}
\EndFor
\State \textbf{return} $length(cascade) \bmod 2 = 1$
\end{algorithmic}
\end{algorithm}

To test if a given ID is valid (i.e., contained in the filter cascade) we must apply Algorithm~\ref{alg:test}. This test depends on the assumption that only IDs from existing \glspl{vc} are tested: $id \in R \cup V$. Otherwise, results will be random. This works because \glspl{rp} only depend on a status check via membership test if they have received a \gls{vc} correctly signed by its issuer. Only then, they will retrieve that issuer's CRSet.
Because Bloom filters can not have false negatives, the index of the first filter that does not contain a given ID determines if the ID is considered a member of the cascade.

Adapting from \cite{larisch2017crlite}, we can express the total cascade size in bits, depending on $n_{max}$, as $5.64 n_{max}$.
We observe that the required size of a filter cascade grows linearly with the capacity.
Using the upper bound introduced in \cite{larisch2017crlite}, the number of levels in the cascade grows logarithmically with the capacity to around 40 levels for several hundred thousand \glspl{vc}. 
This represents the worst case lookup complexity, where every level needs to be tested.
With one hash to compute per level, it is more expensive than looking up one index for \gls{bsl}, but likely still feasible for an \glspl{rp} on modern server hardware.

\subsection{Summary}

The presented architecture for CRSet clearly addresses some initial goals:

\begin{itemize}
    \item Issuers can choose any capacity below the current blob limit per block.
    \item Presentation activity is not trackable because data is fetched from a blockchain with various access paths. Very careful \glspl{rp} could maintain their own node.
    \item Interoperability is achieved through compliance with the \gls{vc} standard. Since the revocation check is non-interactive for the subject, broad interoperability with existing wallet agents and transport protocols is given.
\end{itemize}

\noindent
In the next section, we will further address the remaining goal and show that issuer activity is kept confidential by CRSet.

%% file: content/27_proof.tex
Next, we formalize a strong notion of confidential issuer activity in the context of revocation and provide a proof sketch to show that CRSet achieves it.
The construction parallels game-based proofs from cryptography \cite{shoup2004sequences}.
We define a revocation mechanism through a probabilistic function we name
\begin{equation*}
\mathit{create}: \{ (V,R,n) \in (2^U \times 2^U \times \mathbb{N}) | V \cap R = \emptyset \} \to \{0,1\}^* \text{,}
\end{equation*}
which creates the public revocation data $S_i=\mathit{create}(R_i,V_i,n)$, based on valid and revoked IDs. While a straightforward revocation list would use just $R_i$, other mechanisms may make use of both sets. Capacity-based mechanisms (e.g., \gls{bsl}, CRSet) can make use of $n$, which is a measure of capacity, which can be interpreted according to the mechanism's design.
On $S_i$, we define a membership operator $\in : U \times \{0,1\}^* \to \{true, false\}$. Going forward, we define $d \in S_i$ to mean $d$ is valid.
Any correct \textit{create} function must faithfully represent its input, meaning $d \in V \implies d \in S_i \text{ and } d \in R \implies d \notin S_i$.

Next, we introduce the Chosen Count Indistinguishability Game between a challenger and an adversary $A$ as $CCIG(create, A, l, n)$:

\begin{enumerate}
    \item The adversary $A(create,l,n)$ chooses two distinct, plausible series of counts of valid and revoked \glspl{vc} $N_{V_0},N_{V_1},N_{R_0}, \text{and } N_{R_1}$ over $l$ time steps, where each individual count is smaller than $n$. He sends these to the challenger.
    \item The challenger randomly picks a bit $b$. For every point $0\leq i<l$, he randomly generates sets of $N_{V_b,i}$ valid IDs $V_i$ and $N_{R_b,i}$ revoked IDs $R_i$ in a way that forms a plausible history. Then, for each point $i$, he creates the corresponding revocation data $S_i$ and sends this data to the adversary.
    \item Based on the data, the adversary $A(create,l,n, S_{0}\dots S_{l-1})$ must output his prediction for the used series of counts of IDs as $\hat b$. The result of the game if the adversary wins the game (i.e., $\hat b = b$) is $1$ and otherwise $0$.
\end{enumerate}

\noindent
Plausibility regarding the counts and sets means that nothing may contradict the life cycle of \glspl{vc}. To simplify this slightly for us (and the attacker), we assume no \gls{vc} expiry. Consequently, a \gls{vc} always starts out as valid and stays permanently revoked, if it is revoked.

Based on this game, we can define a more formal notion of confidentiality of issuer activity, encapsulating the idea that an adversary can not accurately compute the absolute or relative number of valid or revoked \glspl{vc} that a given issuer has in circulation at any point in time based on a series of historical revocation data:

\begin{definition}[Chosen Count Indistinguishability] \label{def:countindis}
    A revocation data structure defined by a \texttt{create} function is considered \textbf{chosen count indistinguishable} iff for all efficient algorithms $A$, history lengths $l$, and capacities $n$,
    \begin{equation*}
        | Pr[CCIG(create, A, l, n)=1]- \frac{1}{2} | \text{ is negligible.}
    \end{equation*}
\end{definition}

\noindent
Next, we want to show that CRSet, defined by a create function $crset$, fulfills this definition.
We consider the output $S_i$ of the function to be a serialized format, including cascade and salt.
For this proof sketch, we consider the used hash function to be a random oracle and proceed with a specific version $G_0(A, l, n) = CCIG(crset, A, l, n)$.
For a transition based on indistinguishability \cite{shoup2004sequences}, we introduce a second game $G_1$.
It is based on a modified CRSet \textit{create} function that ignores the input sets and fills the cascade with random padding only: $\mathit{crset}'(R,V,n) = \mathit{crset}(\emptyset,\emptyset,n)$.

\begin{lemma}
    The probability to win the game $G_1$ is $Pr[G_1(A, l, n)=1] = 1/2$.
\end{lemma}
\begin{proof}
Nothing in the output $S$ depends on the hidden bit $b$, because $\mathit{crset}'$ ignores the input sets $V$ and $R$. Therefore, guessing is the best possible strategy, resulting in a win half the time.
\end{proof}

\begin{lemma}
    No efficient adversary $A$ can distinguish between $G_0$ and $G_1$:
    \begin{equation*}
        |Pr[G_0(A, l, n)=1] - Pr[G_1(A, l, n)=1]| \text{ is negligible.}
    \end{equation*}
\end{lemma}
\begin{proof}
We construct the distinguishing game $D(A, l, n, U_D)$ that is a modified version of $G_0$, still with the original $crset$ function. It interpolates between our two games by drawing the sets $V_i$ and $R_i$ from the provided universe $U_D$ instead.
So if the input to $D$ is $U_D=2^U$, the original game $G_0$ is played, and if the input is $U_D=\emptyset$, the modified game $G_1$ is played. Then, the lemma is equivalent to:
\begin{equation*}
    |Pr[D(A, l, n, 2^U)=1] - Pr[D(A, l, n, \emptyset)=1]| \text{ is negligible.}
\end{equation*}

\noindent
Logically, any CRSet revocation data $S_i$ is a set of hashes. Since we assume our hash function $h: \{0,1\}^* \to \{0,1\}^\tau$ to be a random oracle, we know by definition that for any efficient adversary $A$
\begin{equation*}
\begin{split}
    | &Pr[s\leftr \{0,1\}^*, d\leftr \{0,1\}^*: A(s,h(d||s))=1]\\
    -&Pr[s\leftr \{0,1\}^*, r\leftr \{0,1\}^\tau: A(s,r)=1] |
\end{split}
\end{equation*}

\noindent
is negligible, where $\leftr$ signifies random sampling from a distribution.
For two possible input distributions $X,Y$ over $\{0,1\}^*$, after applying the triangle inequality, this implies
\begin{equation*}
\begin{split}
    | &Pr[s\leftr \{0,1\}^*, x\leftr X: A(s,h(x||s))=1]\\
    - &Pr[s\leftr \{0,1\}^*, y\leftr Y: A(s,h(y||s))=1] | \leq 2\epsilon_{RO} \textbf{,}
\end{split}
\end{equation*}
where $\epsilon_{RO}$, and thus also $2\epsilon_{RO}$, is a negligible advantage.
That implies that the hash sets produced in the game $D$ cannot be used to distinguish the different input sets, which correspond to our two games $G_0$ and $G_1$.
\end{proof}

\begin{theorem}
    CRSet is \textbf{chosen count indistinguishable}.
\end{theorem}

\begin{proof}
From combining the two preceding lemmas, it follows that
\begin{equation*}
    |Pr[G_0(A, l, n)=1] - 1/2| \text{ is negligible,}
\end{equation*}
which matches Definition~\ref{def:countindis}.
\end{proof}

%% file: content/30_eval.tex
We evaluate CRSet through both a core implementation for privacy and performance analysis and a more deployment-ready end-to-end prototype. Finally, we acknowledge practical limitations.

\subsection{Empirical Core Evaluation}

We conduct several tests to back up our theoretical results. For this purpose, we built a core implementation in Python which we used to perform experiments for privacy and performance, respectively.
Our core implementation\footnote{\url{https://github.com/flhps/crset-empirical-experiments}} builds on a modified version\footnote{\url{https://github.com/flhps/rbloom}} of the \texttt{rbloom} library, using \texttt{sha256} as the hash function. We experimentally determined the optimal false positive rate as $p=0.53$, close to the theoretical optimum of $p=1/2$. To ensure reliable creation, we initialize Bloom filters for at least $1024$ entries, preventing unexpectedly high false positive rates on the final cascade levels.

Beyond our formal privacy guarantees, we conduct empirical machine learning experiments to observe privacy preservation in practice. These experiments demonstrate resistance against statistical inference attacks and provide complementary evidence of CRSet's privacy properties.
Our experiments employ supervised learning to assess if an adversary could predict revoked \glspl{vc} counts from cascade structural properties. We generate two million-entry datasets (with and without padding) using a cascade capacity of approximately $\hat r = 10^5$. Features include total size, filter count, sizes of the first three filters, and their set bit counts, normalized with standard scaling. We use an 80-20 train-test split with ridge and lasso regression, measuring performance through \gls{mse} and R².
This approach seems to confirm CRSet's privacy protection. For padded cascades, regression models achieve R² values near 0, indicating structural features reveal no information about revocation numbers. This is supported by \gls{mse} values comparable to the underlying data variance, with models defaulting to near-average predictions regardless of input features. In contrast, unpadded cascades readily reveal \gls{vc} counts, confirming padding's essential role. Similar results occur when predicting valid \glspl{vc} counts.

\begin{figure}[t]
    \centering
    \includegraphics[width=\textwidth]{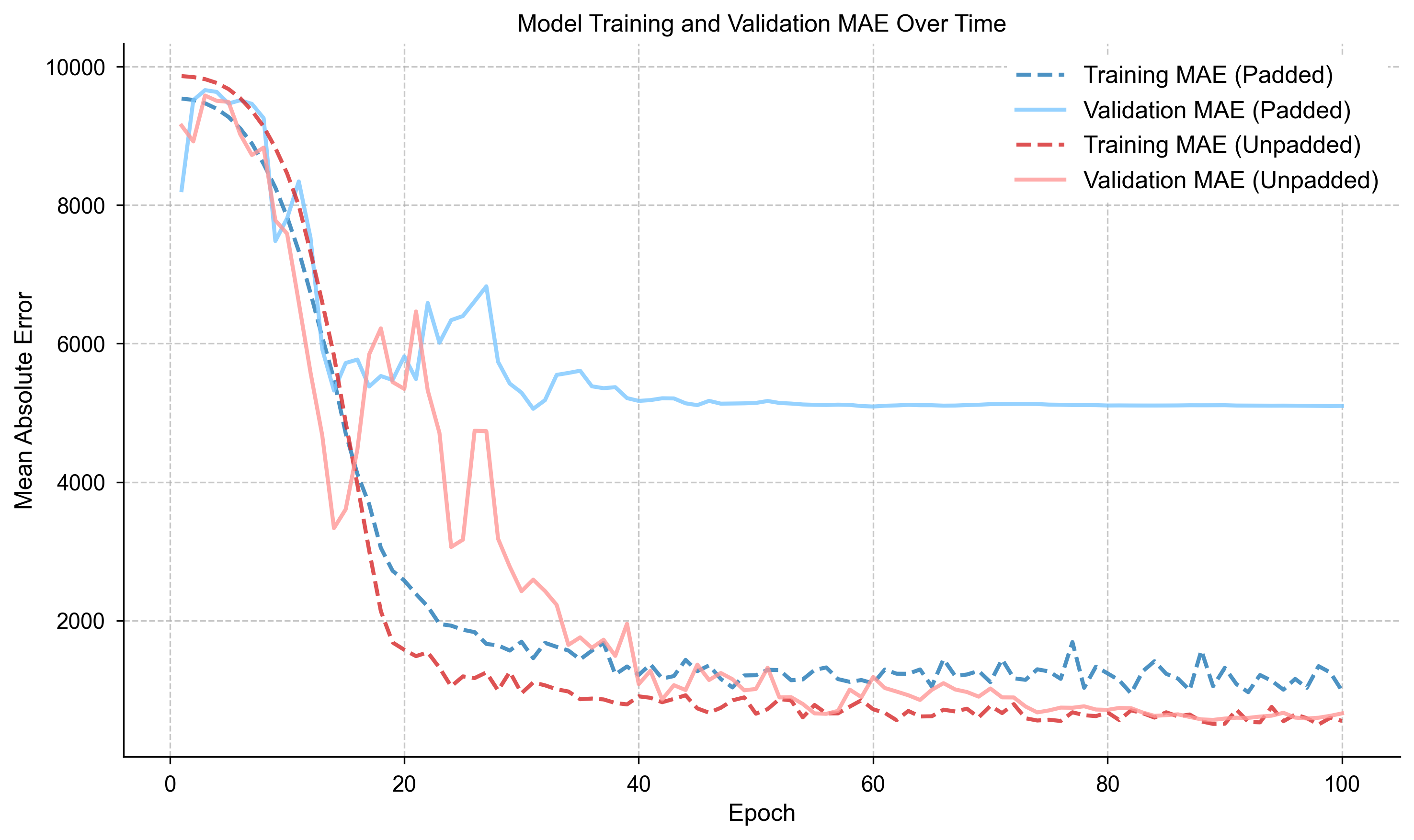}
    \caption{Model Training Without Manual Feature Engineering and Validation \gls{mae} Over Time for Padded and Unpadded Datasets}
    \label{fig:training}
\end{figure}

To validate our privacy mechanism without manual feature engineering, we also evaluate two uniformly distributed datasets of 10,000 data points each (with and without padding) with mean counts of close to 10,000 valid \glspl{vc}. These values are chosen to still be able to perform experiments on consumer hardware. We use a deep neural network with multiple dense layers, dropout, and batch normalization, using an 80-10-10 train-validation-test split and \gls{mae} as our metric.
Figure~\ref{fig:training} reveals distinct behaviors between padded and unpadded cascades. For unpadded cascades, both training and validation \gls{mae} converge to zero, confirming the network can extract \gls{vc} counts from binary representations. However, with CRSet's padded cascades, while training error decreases (showing memorization), validation error stabilizes around $5,000$ after starting near $10,000$ (random prediction), precisely matching the expected absolute deviation for a uniform distribution. This convergence to theoretical random prediction error empirically supports CRSet's privacy-preserving characteristics.

\begin{figure}[t]
    \centering
    \includegraphics[width=\textwidth]{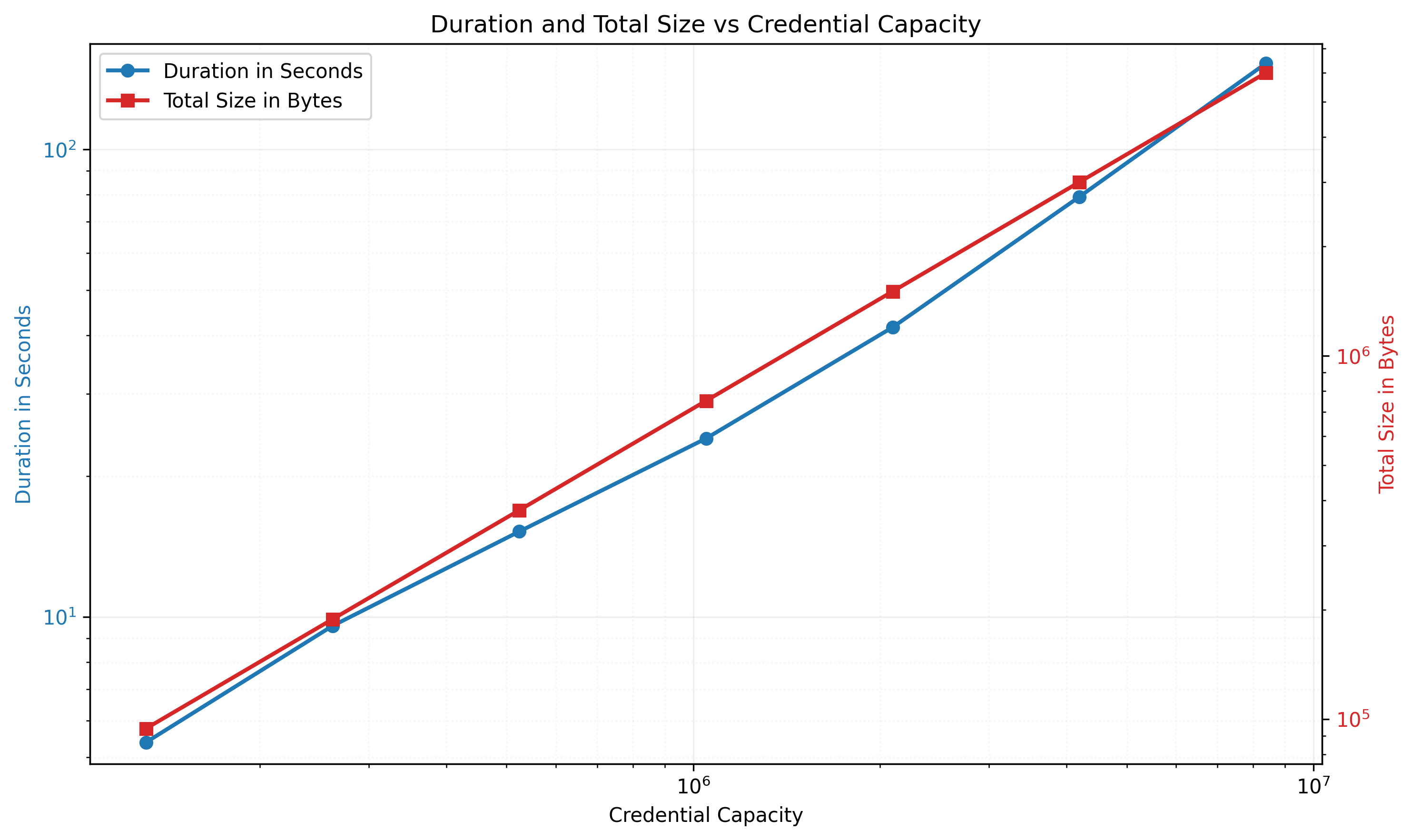}
    \caption{Creation Duration and Total Size of CRSets Depending on Credential Capacity}
    \label{fig:performance}
\end{figure}

We conducted performance evaluations on a system with an AMD Ryzen 7 Pro 7840U processor and 32 GB of RAM. Our analysis reveals that both creation duration and total size exhibit near-linear growth in logarithmic space as capacity increases. Duration scales predictably from 5 seconds at $10^2$ \glspl{vc} to 200 seconds at $10^7$ \glspl{vc}, with total size scaling proportionally from $10^5$ to $10^7$ bytes. These consistent scaling characteristics enable accurate resource requirement predictions based on anticipated VC volume.
While creation failures are possible, our reliability testing across 100,000 creation attempts observed none.

\subsection{End-to-End Evaluation}

We implemented the full CRSet mechanism in JavaScript for web-based projects, using Ethereum for decentralized publishing. Our implementation comprises:

\begin{itemize}
    \item \textit{CRSet Cascade}: A pure JavaScript library implementing our padded cascade, using \texttt{sha256} as a robust cryptographic hash function.
    \item \textit{CRSet Check (CRSC)}: A library enabling CRSet revocation checks for \glspl{rp}.
    \item \textit{CRSet Issuer Backend (CRSIB)}: A standalone backend service for issuers to manage their CRSet, including revocation ID tracking and CRSet updates.
    \item \textit{CRSet Demonstrator}\footnote{\url{https://github.com/flhps/crset-demo}}: A demonstration scenario making use of the previous components, featuring mock web applications for an issuer and an \gls{rp}, which interact using \gls{oid4vc} protocols \cite{oid4vc}.
\end{itemize}

The issuer backend provides a REST API allowing issuers to create credential status entries that can be written directly into \gls{vc} payloads before signing. These entries are stored in the backend's database, with additional endpoints for listing statuses, revoking \glspl{vc} by ID, and triggering CRSet creation and publishing via blob transactions.
Publishing costs fluctuate heavily with blob gas fees (2-50 Gwei in late 2024 and early 2025\footnote{\url{https://blobscan.com/stats}}), resulting in per-blob costs ranging from cents to approximately 57 USD (at 1 ETH = 3574 USD).
However, the Pectra upgrade has increased the possible number of blobs per block \cite{eip7691} and quickly brought prices for one blob down to cents.

For verification, \glspl{rp} first use standard \gls{vc} verifier software for format and signature checks, then call the CRSC to verify revocation status. From the address in the \texttt{statusEntry.id} field, the CRSC retrieves the newest blob-carrying transaction via blockchain RPC (using Moralis\footnote{\url{https://developers.moralis.com/}}), which yields one or more blob hashes \cite{eip4844}. Since blob data resides on the consensus layer, the library fetches it from a consensus node (via Blobscan\footnote{\url{https://api.blobscan.com/}}), concatenates and deserializes the blobs into a filter cascade, and checks if the revocation ID (the last element in the \texttt{statusEntry.id} field) is included. Verification performance is primarily constrained by these two network round-trips.

\subsection{Practical Limitations}

CRSet's design has practical limitations. Since it is a hash-based mechanism and the data is public, long-term privacy cannot be guaranteed. With time, an adversary may eventually be able to extract information by brute-force.
Also, like any other capacity-based revocation mechanism, CRSet implicitly telegraphs that a prior instance was likely filled up whenever a new one is created. However, this is only detectable by seeing a new \gls{vc} using the new instance (i.e., new Ethereum address), because addresses are not inherently linked.
In many cases, instances can likely be sized to last decades, severely limiting leaked information.

The blockchain-based publishing approach also has practical challenges.
The operational complexity of accessing blockchain blob data could limit the confidentiality of presentation activity, for example through the use of third-party API providers.
Also, while the current prices for blobs are very low, there is no hard upper bound because they are directly driven by the exchange rate of ether and the blob demand.

%% file: content/40_conclusion.tex
In this paper, we identified the privacy of an issuer's activity as an important yet underexplored requirement in SSI. To address this, we introduced CRSet, a non-interactive revocation mechanism based on padded Bloom filter cascades. We formalized a privacy guarantee using a game-based security model and provide proof demonstrating that CRSet achieves chosen count indistinguishability. Our empirical evaluation confirms that CRSet performs reliably at scale. While the mechanism is technically deployable on existing blockchain infrastructure via blob-carrying transactions, the lack of upper bound for the associated cost remains a challenge for productive adoption.